\documentclass[lettersize, 10 pt, journal, twoside]{IEEEtran}  % Comment this line out

\title{Sampling-based Reactive Synthesis for Nondeterministic Hybrid Systems}
\author{Qi Heng Ho, Zachary N. Sunberg, and Morteza Lahijanian%
\thanks{Manuscript received: September, 13, 2023; Accepted November, 16, 2023.}%Use only for final RAL version
\thanks{This paper was recommended for publication by Editor Aniket Bera upon evaluation of the Associate Editor and Reviewers' comments.
This work was supported by Strategic University Research Partnership (SURP) grants from the NASA Jet Propulsion Laboratory (JPL) (RSA 1688009 and 1704147).} %Use only for final RAL version
\thanks{The authors are with the Department of Aerospace Engineering Sciences, University of Colorado Boulder, CO, USA
        {\tt\small \{\textit{firstname}.\textit{lastname}\}@colorado.edu}}
\thanks{Digital Object Identifier (DOI): 10.1109/LRA.2023.3340029}
}

\let\proof\relax
\let\endproof\relax
\usepackage{amsthm}
\usepackage[utf8]{inputenc}
\usepackage{amsmath,amssymb}
\usepackage{booktabs}
\usepackage{graphicx}
\usepackage[font=small]{subcaption}
\usepackage[dvipsnames]{xcolor}
\usepackage{numprint}
\usepackage{comment}
\usepackage[linesnumbered,ruled,noend]{algorithm2e}
\usepackage{eso-pic}
\usepackage[american]{babel}
\usepackage{svg}
\usepackage{xspace}

\usepackage{adjustbox}
\usepackage{tikz}
\usetikzlibrary{automata,positioning,shapes,arrows}

\usepackage[normalem]{ulem} % only for \zsedit below

\usepackage{tabularx,siunitx}
\usepackage{cite}
\usepackage{hyperref}
\usepackage{multirow}

\newtheorem{theorem}{Theorem}
\newtheorem{problem}{Problem}
\newtheorem{example}{Example}
\newtheorem{assumption}{Assumption}

\newtheorem{lemma}{Lemma}
\newtheorem{definition}{Definition}
\newtheorem{remark}{Remark}

\DeclareMathOperator*{\argmin}{argmin}

\newcommand{\N}{\mathcal{N}}
\newcommand{\E}{\mathcal{E}}

\newcommand{\tree}{\mathcal{T}}
\newcommand{\strategytree}{\pi}

\newcommand{\modesetprime}{Q'}

\newcommand{\andortree}{\mathcal{T}}
\newcommand{\subtree}{\andortree_{sub}}
\newcommand{\UCB}{\textsc{ucb}}
\newcommand{\Q}{\mathcal{Q}}
\newcommand{\SaBRS}{{SaBRS}\xspace}
\begin{document}
\maketitle
\begin{abstract}
    This paper introduces a sampling-based strategy synthesis algorithm for nondeterministic hybrid systems with complex continuous dynamics under temporal and reachability constraints. We model the evolution of the hybrid system as a two-player game, where the nondeterminism is an adversarial player whose objective is to prevent achieving temporal and reachability goals. The aim is to synthesize a winning strategy -- a reactive (robust) strategy that guarantees the satisfaction of the goals under all possible moves of the adversarial player. Our proposed approach involves growing a (search) game-tree in the hybrid space by combining sampling-based motion planning with a novel bandit-based technique to select and improve on partial strategies. We show that the algorithm is probabilistically complete, i.e., the algorithm will asymptotically almost surely find a winning strategy, if one exists. The case studies and benchmark results show that our algorithm is general and effective, and consistently outperforms state of the art algorithms.
\end{abstract}
\section{Introduction}
\IEEEPARstart{R}{eactive} synthesis is the problem of generating a control strategy that enables a robot to \emph{react} to uncertainties on the fly to guarantee satisfaction of complex requirements. The requirements are often expressed in \textit{temporal logic} (TL) such as \textit{linear TL} (LTL) \cite{Baier2008} for specification on the sequence of events and \textit{metric interval TL} (MITL) and \emph{signal TL} (STL) \cite{maler2004monitoring} for dense-time specifications. Although reactive synthesis is known to be hard, it is an active area of research due to its applications in \emph{safety-critical} and \emph{time-critical} systems such as autonomous driving, search-and-rescue, and surgical robotics \cite{kress2018synthesis}. Reactive synthesis is often studied in the discrete setting, where the dynamics are abstracted to a finite model \cite{He2017reactive, Muvvala2022collaborate, Vasile2020reactivesynthesis, hadas2012reactive}. For complex and uncertain dynamics with dense-time requirements, however, such abstractions are either unavailable or too coarse (in both space and time), preventing accurate analysis and completeness guarantees. This work focuses on reactive synthesis for such systems and aims to develop an algorithm with correctness and completeness guarantees.

A powerful and expressive model that represents complex robotic systems under uncertainty with TL specifications is \textit{Nondeterministic Hybrid Systems} (NHS) \cite{NHS}. NHS allows both continuous and discrete dynamics via discrete modes that contain continuous dynamics and discrete switching between the modes. An NHS can be viewed as the composition of a robot's continuous dynamics with its environment and TL requirements, where the continuous dynamics and dense-time requirements are captured within each discrete mode, and switching between modes either represents changes in the dynamics or environment, or capture the requirements on the sequence of events. In this view, the satisfaction of the requirements reduces to a reachability objective for the NHS, and hence, the problem becomes synthesis of a strategy that guarantees reachability under all uncertainties.

\begin{example}
\label{ex:searchandrescue}
Consider a search-and-rescue scenario, where a building is on fire, in which there may be a trapped human that needs to be rescued, as depicted in Fig.~\ref{fig:searchandrescue}.
To aid the search for the human, a rover with second-order car dynamics is tasked with searching and mapping every room of the building within $2$ minutes. If a room is unblocked, the rover must search it in $10$s. If the human is found, the robot must protect the human until the rescue team arrives. If all rooms are found to be blocked or no human is found, the rover must go to the exit zone in green within $20$s to report the map to the rescue team. In this example, the uncertainty is in the environment, where the rooms may or may not be blocked or contain a human.
\end{example}

% \begin{figure}[t]
%     \centering
%     \begin{subfigure}{0.35\textwidth}
%     \includegraphics[width=\linewidth, height=0.162\textheight]{figures/robot_experiment_image.png}
%     \end{subfigure}
%     \begin{subfigure}{0.121\textwidth}
%     \includegraphics[width=\linewidth, trim={19.7cm 2.8cm 18.1cm 2.0cm},clip]{figures/robot_experiment_strategy.pdf}
%     \end{subfigure}
% \caption{\small Search-and-rescue scenario. Left: a robot is navigating a building with unknown room blockage and human states. Right: computed robot strategy that reacts to all possibilities of room blockages and human presence in each room. Each colored trajectory is a possible evolution of the reactive strategy.}
% \label{fig:searchandrescue}
% \vspace{-4.3mm}
% \end{figure}

\begin{figure}[t]
    \centering
    \begin{subfigure}{0.35\textwidth}
    \includegraphics[width=\linewidth, height=0.162\textheight]{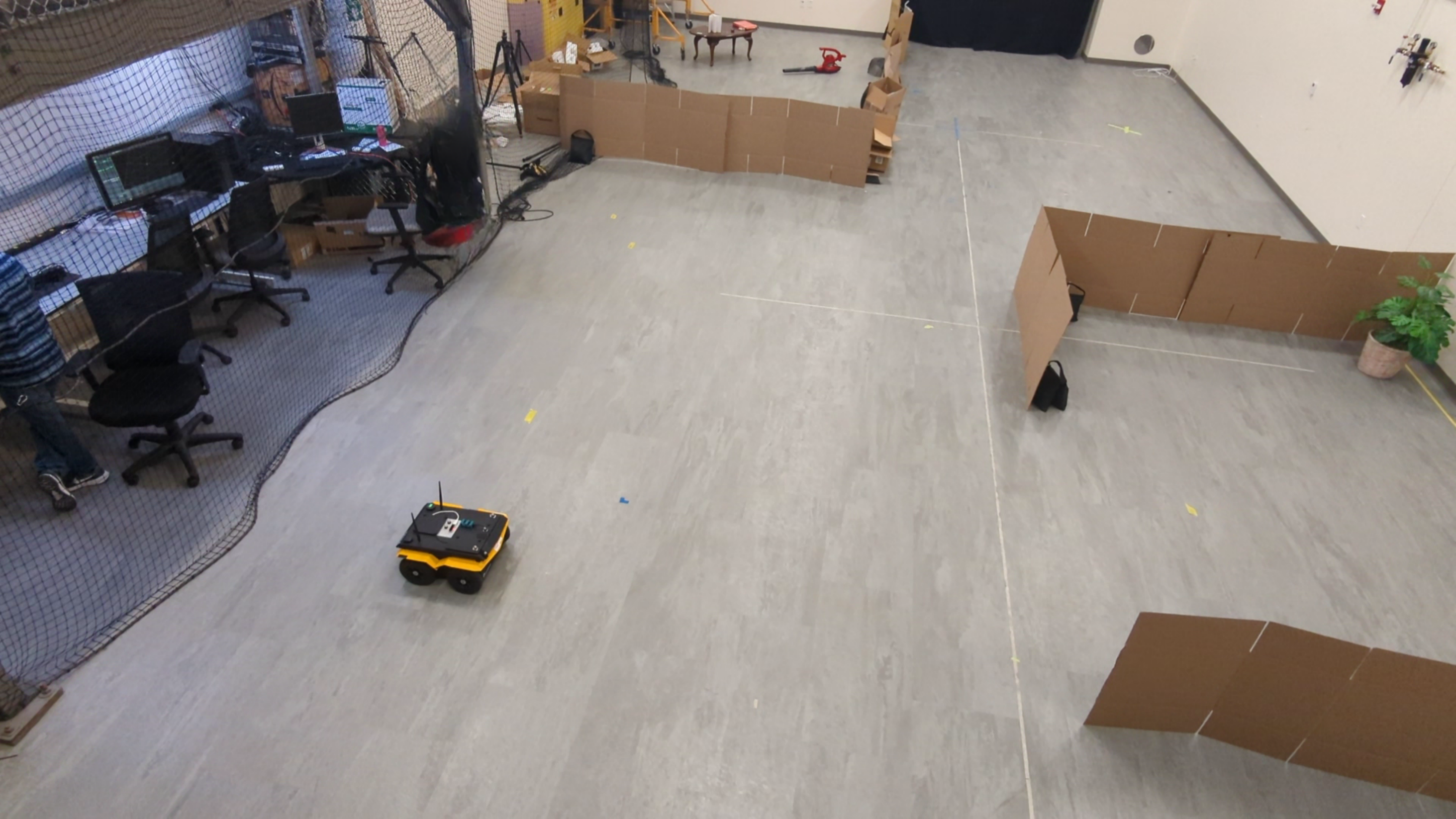}
    \end{subfigure}
    \begin{subfigure}{0.121\textwidth}
    \includegraphics[width=\linewidth, trim={19.7cm 2.8cm 18.1cm 2.0cm},clip]{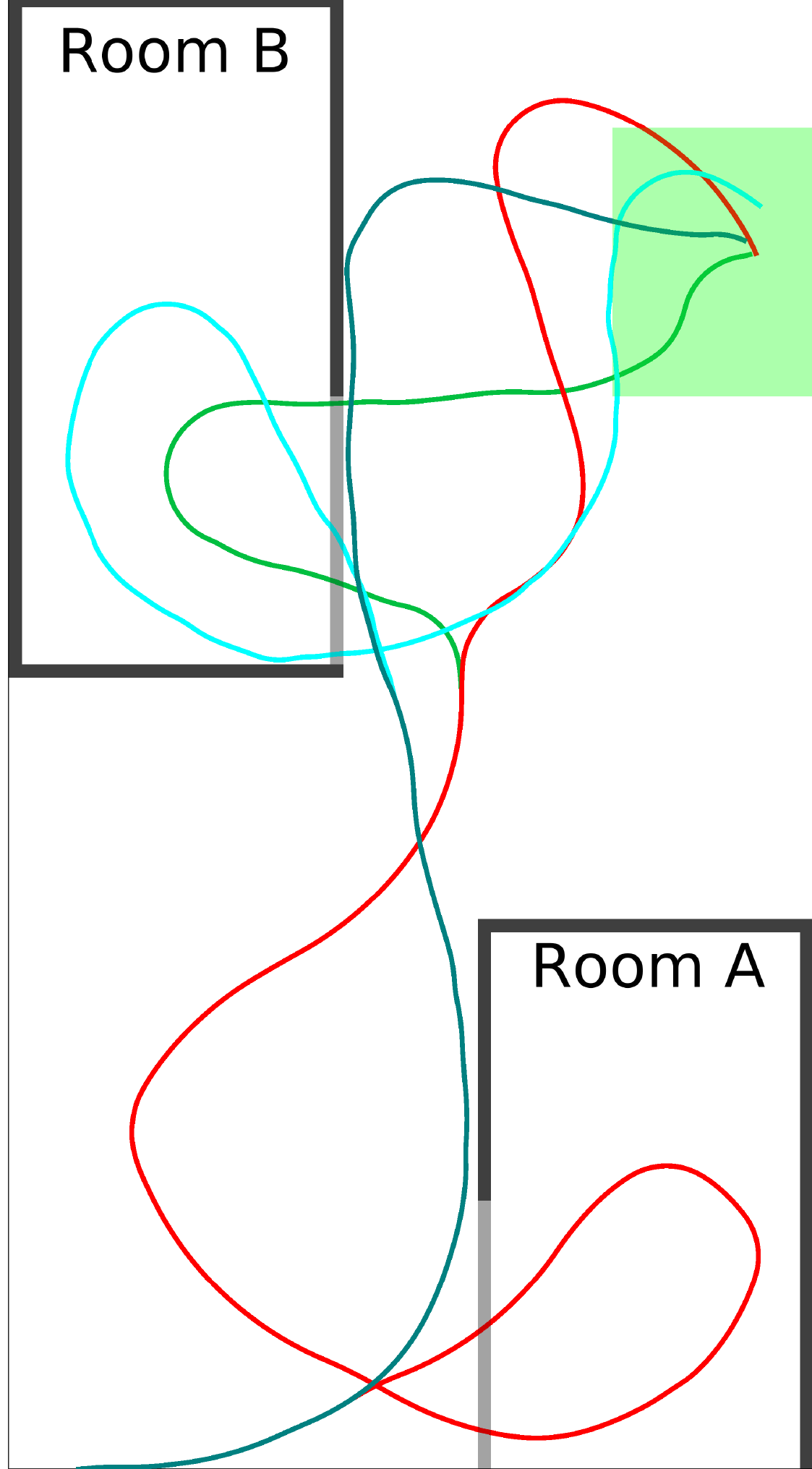}
    \end{subfigure}
\caption{\small Search-and-rescue scenario. Left: a robot is navigating a building with unknown room blockage and human states. Right: computed robot strategy that reacts to all possibilities of room blockages and human presence in each room. Each colored trajectory is a possible evolution of the reactive strategy.}
\label{fig:searchandrescue}
\vspace{-4.3mm}
\end{figure}

Our approach is based on a game-theoretic interpretation of the problem.  We model nondeterminism as an adversarial player that attempts to prevent the robot from achieving its temporal and reachability objectives. This game is in the hybrid space, which is infinite and uncountable. Therefore, finite game techniques that are common in abstraction-based approaches are not applicable here.  Instead, we aim to synthesize a strategy directly in the hybrid state space by iteratively constructing a game tree and exploring ``promising'' strategies.  This however poses two main challenges: (i) construction of the game tree with nonlinear continuous dynamics and (ii) the \emph{exploration-exploitation} dilemma.  To deal with challenge (i), we take inspirations from the tremendous success of sampling-based techniques in motion planning.  To overcome challenge (ii), we adapt multi-armed bandit methods developed for planning under uncertainty problems. We devise an algorithm, called 
\emph{Sampling-based Bandit-guided Reactive Synthesis} (\SaBRS), that uses a novel bandit-based method to select a strategy in the game tree for expansion and employs random sampling to grow this strategy. We show that the algorithm is probabilistically complete, i.e., the algorithm asymptotically almost surely finds a strategy that guarantees satisfaction of the objectives if one exists.

The contributions of this paper are threefold: (i) \SaBRS, a novel sampling-based reactive synthesis algorithm for NHS, (ii) proof of probabilistic completeness of \SaBRS, (iii) a series of benchmarks and case studies, including real robot demonstrations, that illustrate the generality and efficacy of \SaBRS. Results show that \SaBRS consistently finds solutions up to an order of magnitude faster than the state of the art. To the best of our knowledge, this is the \emph{first} probabilistically-complete reactive synthesis algorithm for NHS.
\subsection{Related Work}
Sampling-based algorithms have emerged as powerful tools for kinodynamic motion planning for robotic systems with complex temporal goals and nonlinear or hybrid dynamics \cite{Kingston2018samplingbased, kinoRRT, Bhatia2010, HyRRT}. These techniques are typically used for deterministic systems and, only recently, extended to stochastic models \cite{CC-RRT, Ho2022gbt, Ho2023simba, Theurkauf2023ETmotionplanning, agha2014firm}. Nondeterminism, where no probability distributions are available, is often not considered. In this work, we utilize sampling-based methods to achieve reachability objectives for nonlinear hybrid systems with nondeterminism.

A common approach to handle nondeterminism is to model it as an adversarial player in a game setting.  Reactive synthesis is based on this view and typically studied in discrete games \cite{kress2018synthesis, He2017reactive, Muvvala2022collaborate}. 
When applied to continuous systems, however, they require finite abstraction, which is difficult to obtain for complex dynamics. In the continuous domain, techniques based on Hamilton-Jacobi analysis, contraction theory, and counterexample-guided synthesis have been employed to provide robust controllers with guarantees on system behavior \cite{fastrack, hybridgames, Singh2017contractiontheory ,Tsukamoto2021contractiontheory, Raman2015reactivesynthesis}. However, these methods are designed for continuous disturbance and cannot handle discrete nondeterminism. In this work, we formulate the problem as a two-player minimax game directly in the hybrid space and propose an algorithm for efficient reactive synthesis with formal guarantees.

The work that most closely relates to ours is \cite{Lahijanian2014}. It considers the same problem and proposes a two-phase sampling-based strategy planner that performs exploration in the first phase and strategy improvement in the second. The algorithm is highly dependent on the quality of strategies found in the first phase, and since it cannot return to the first phase, it is incomplete. In this work, we develop a probabilistically complete algorithm that continually improves strategies.
\section{Problem Formulation}

In this work, we consider complex robotic systems under uncertainty with temporal and reachability objectives. Specifically, we focus on uncertanties due to nondeterministism or discrete disturbances. A general modeling framework that allows for accurate representation of such systems and objectives simultaneously is nondeterministic hybrid systems.

\begin{definition}[NHS] 
    \label{def: NHS}
    A \emph{nondeterministic hybrid system} (NHS) is a tuple $H = (S, s_0, U, I, F, E, G, J, S_{goal}, R)$, where
    \begin{itemize}
        \item $S = Q \times X$ is the hybrid state space which is the Cartesian product of a finite set of discrete modes, $Q = \{q_1, q_2, \ldots, q_m\}$ for $m \in \mathbb{N}$, with a set of mode-dependent continuous state spaces $X = \{X_q \subseteq \mathbb{R}^{n_q} \mid q \in Q \}$, where $n_q \in \mathbb{N}$,
        \item $s_0 \in S$ is the initial state,
        \item $U = \{U_q \subset \mathbb{R}^{m_q} \mid q \in Q \}$, where $m_q \in \mathbb{N}$, is the set of mode-dependent control spaces,
        \item $I = \{I_q : X_q \rightarrow \{\top, \bot\} \mid q \in Q\}$ is the set of invariant functions, 
        \item $F = \{F_q : X_q \times U_q \rightarrow X_q \mid q \in Q\}$ is the set of flow functions that describe the continuous dynamics of the robot in each mode,
        \item $E \subseteq Q \times Q$ is the discrete transitions between modes,
        \item $G = \{G_{q \modesetprime} : X_q \rightarrow \{\top, \bot\} \mid \modesetprime \in 2^Q, \; \forall q' \in \modesetprime, (q, q') \in E \}$ is a set of guard functions that, given hybrid state $(q,x)$, $G_{q \modesetprime}(x) = \top$ triggers a transition from mode $q$ to a mode in $\modesetprime$. If
        $|\modesetprime| = 1$, the transition is deterministic; otherwise, it is nondeterministic,  
        \item $J = \{J_{q q'} : X_q \rightarrow X_{q'} \mid (q, q') \in E\}$ is a set of jump functions that, once a transition $(q,q')$ is triggered by a guard at state $x \in X_q$, $J_{qq'}(x) \in X_{q'}$ resets the continuous state in mode $q'$, 
        \item $S_{goal} \subseteq S$ is a set of goal states that define the reachability objective, and
        \item $R : S \rightarrow \{\top, \bot\}$ is the reachability indicator function, where $R(s) = \top$ if $s \in S_{goal}$, and $R(s) = \bot$ otherwise.
    \end{itemize}
\end{definition}

\noindent
In this definition of NHS, nondeterminism is in the discrete transition.  
The temporal constraints are typically encoded in the invariant $I$ and guard $G$ functions, and the reachability objective is explicitly defined by set $S_{goal}$ and its indicator function $R$.
The evolution of the NHS is determined by a control strategy, which picks control actions for the system.

\begin{figure}[t]
    \centering
    \includegraphics[width=0.75\linewidth, , trim={0cm 3cm 0cm 0cm},clip]{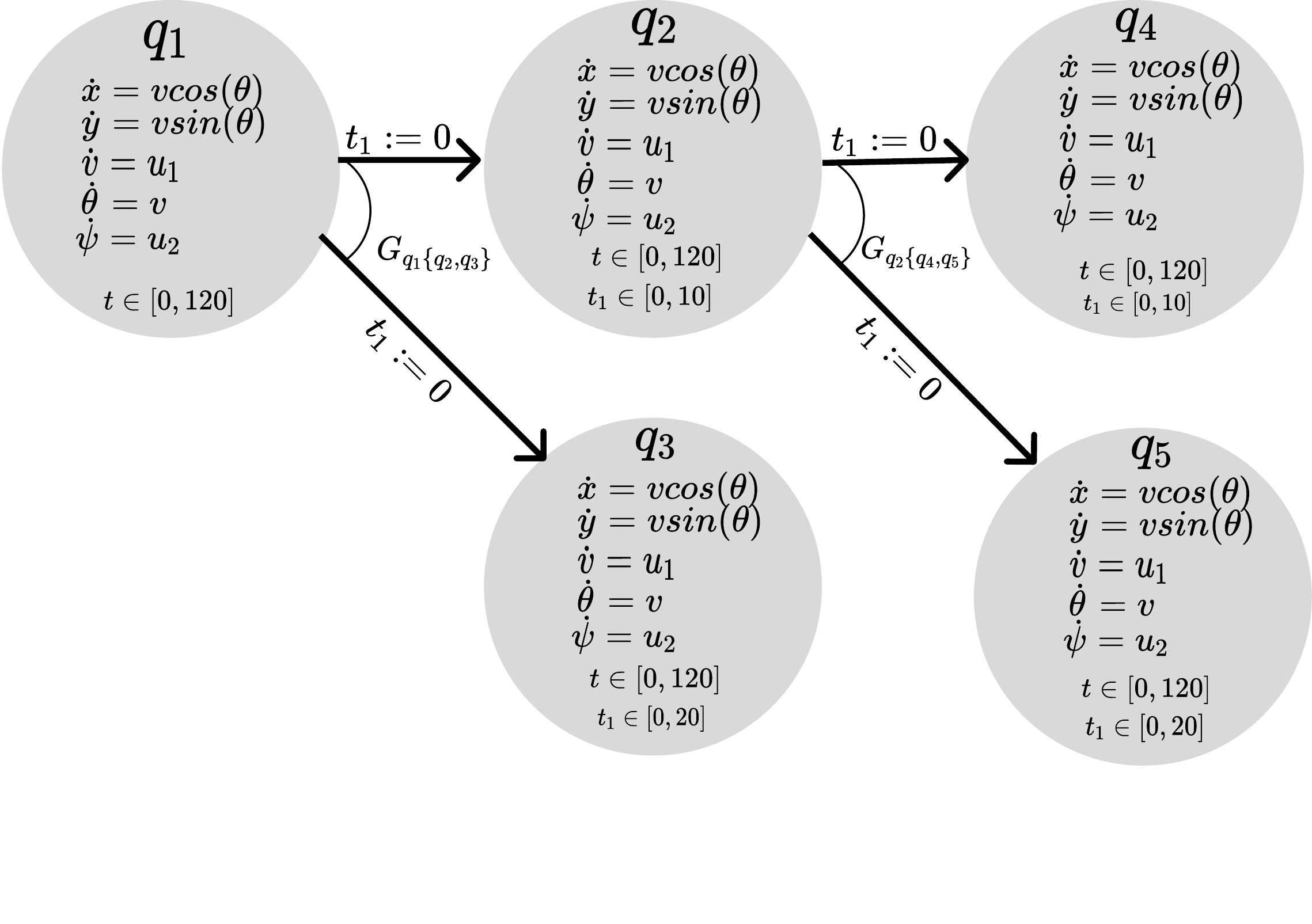}
    \caption{\small NHS for single room version of Example~\ref{ex:NHS model}.}
    \label{fig:searchrescueNHS}
    \vspace{-4.5mm}
\end{figure}

\begin{definition}[Control Strategy]
    A control strategy $\pi :  S \rightarrow \cup_{q\in Q} U_q$ is a function that, for hybrid state $s = (q,x)$, chooses an input control $u \in U_q$.
\end{definition}

Under control strategy $\pi$, the evolution of $H$ is as follows. From initial state $s_0 = (q_0, x_0)$, the continuous component, $x_t$, of hybrid state $s_t=(q_0, x_t)$ evolves according to dynamics $\dot{x} = F_q(x_t, \pi(q_0,x_t))$ until a guard in mode $q_0$ is triggered. Let $\tau$ denote the time that the system first hits the guard, i.e., $G_{q_0 \modesetprime}(x_{\tau}^-) = \top$. Then, the system's discrete dynamics (mode) makes a transition to $q' \in \modesetprime$ nondeterministically, and the continuous state is updated according to the jump function, i.e., $x_\tau^+ = J_{q_0 q'}(x_\tau^-)$.
In mode $q'$, the system's continuous component evolves according to flow $F_{q'}$ from $x_\tau^+$. This process continues until either the invariant function $I$ becomes false, which indicates that temporal constraints are violated, or the reachability indicator function $R$ becomes true, which indicates that the reachability objective is satisfied.  

Note that the NHS in Def.~\ref{def: NHS} assumes deterministic continuous dynamics in every mode. It is the discrete dynamics (switching between the modes) that is subject to uncertainty, i.e., nondeterministic transitions.

\begin{example}
\label{ex:NHS model}
The NHS that models a simplified version of Example~\ref{ex:searchandrescue} (Fig.~\ref{fig:searchandrescue}) is shown in Fig.~\ref{fig:searchrescueNHS}.  Time parameters $t$ (global clock) and $t_1$ (local clock) are added to the continuous state $x$.  Note that $t_1$ resets at every discrete transition; hence, the temporal constraints on the system are captured in the invariant $I_{q_{\cdot}}$ in each mode. The positions that enable the robot to observe the status of the room door in mode $q_1$ represents a guard region that triggers a nondeterministic transition for closed (mode $q_3$) or open (mode $q_2$) door. Searching the room represents a guard region that transitions the system to the next mode with a trapped human (mode $q_4$) or no human found (mode $q_5$). 
 These are nondeterministic guards because the status of the door/room is unknown. 
By reaching the green region in mode $q_3$ and $q_5$ or finding a human in mode $q_4$, $R$ becomes true, which satisfies the timed reachability objective.
\end{example}

To guarantee existence and uniqueness of solution (trajectory) in each mode and enable completeness analysis, we assume the flow and jump functions are Lipschitz continuous.

\begin{assumption}[Lipschitz Continuity]
    \label{assumption: Lipschitz}
    For every mode $q \in Q$, flow function $F_q$ is Lipschitz continuous in both continuous state and control, i.e., there exists constants $L_x, L_u > 0$ such that  $\forall x_1, x_2 \in X_q$ and $\forall u_1, u_2 \in U_q$,
    \begin{align*}
        \|F_q(x_1, u_1) - F_q(x_2, u_2) \| \leq L_x \|x_1 - x_2 \| + L_u \|u_1 - u_2\|.
    \end{align*}
    Further, for every transition $(q,q') \in E$, jump function $J_{qq'}$ is Lipschitz continuous in the continuous state, i.e., $\forall x_1,x_2 \in X_q$, there exists a constant $K_x > 0$ such that
    \begin{align*}
        \|J_{q q'}(x_1) - J_{q q'}(x_2) \| \leq K_x \|x_1 - x_2 \|.
    \end{align*}
\end{assumption}

While Assumption~\ref{assumption: Lipschitz} guarantees that the continuous state trajectories are unique in each mode given a control strategy $\pi$, multiple hybrid state trajectories are still possible due to nondeterminism in the guards, i.e., nondeterminism in the discrete dynamics. In this work, we seek control strategies that are robust to these nondeterministic possibilities.
That is, the control strategy guarantees the completion of reachability and temporal objectives by considering all possible outcomes. Such strategies are called winning.

\begin{definition}[Winning Control Strategy]
    For NHS $H$,
    control strategy $\pi^*$ is \emph{winning} if every hybrid state trajectory induced by $\pi^*$ terminates in $S_{goal}$.
\end{definition}
\noindent
In this work, our goal is to find a winning control strategy. 

\begin{problem}[Reactive Synthesis]
    \label{prob:winning strategy}
    Given a nondeterministic hybrid system $H$ as in Def.~\ref{def: NHS} with initial state $s_0$ and a goal set $S_{goal} \subseteq S$, synthesize a winning control strategy $\pi^*$ that guarantees reaching $S_{goal}$ from $s_0$.
\end{problem}

\begin{remark}
    The (robust) hybrid system reachability problem formulated in Problem~\ref{prob:winning strategy} captures a large class of uncertain systems with (finite) TL (e.g., LTLf \cite{ltlf}, co-safe LTL \cite{Kupferman1999ModelCO}, and MITL and STL \cite{maler2004monitoring}) objectives, where the hybrid system is the Cartesian product of an uncertain continuous system with the automaton that represents the TL objectives.
\end{remark}

% \begin{figure*}[ht]
%     \centering
%     \includesvg[width=0.9\linewidth]
%     {figures/UCST.svg}
%     \caption{ Illustration of SaBRS algorithm. Circles represent \texttt{OR} (\texttt{MIN} player) nodes, while squares represent \texttt{AND} (\texttt{MAX} player) nodes.}
%     \label{fig:framework}
%     \vspace{-4.5mm}
% \end{figure*}

\begin{figure*}[ht]
    \centering
    \includegraphics[width=\linewidth, trim={0.02cm 2.72cm 0.02cm 0.6cm}, clip]{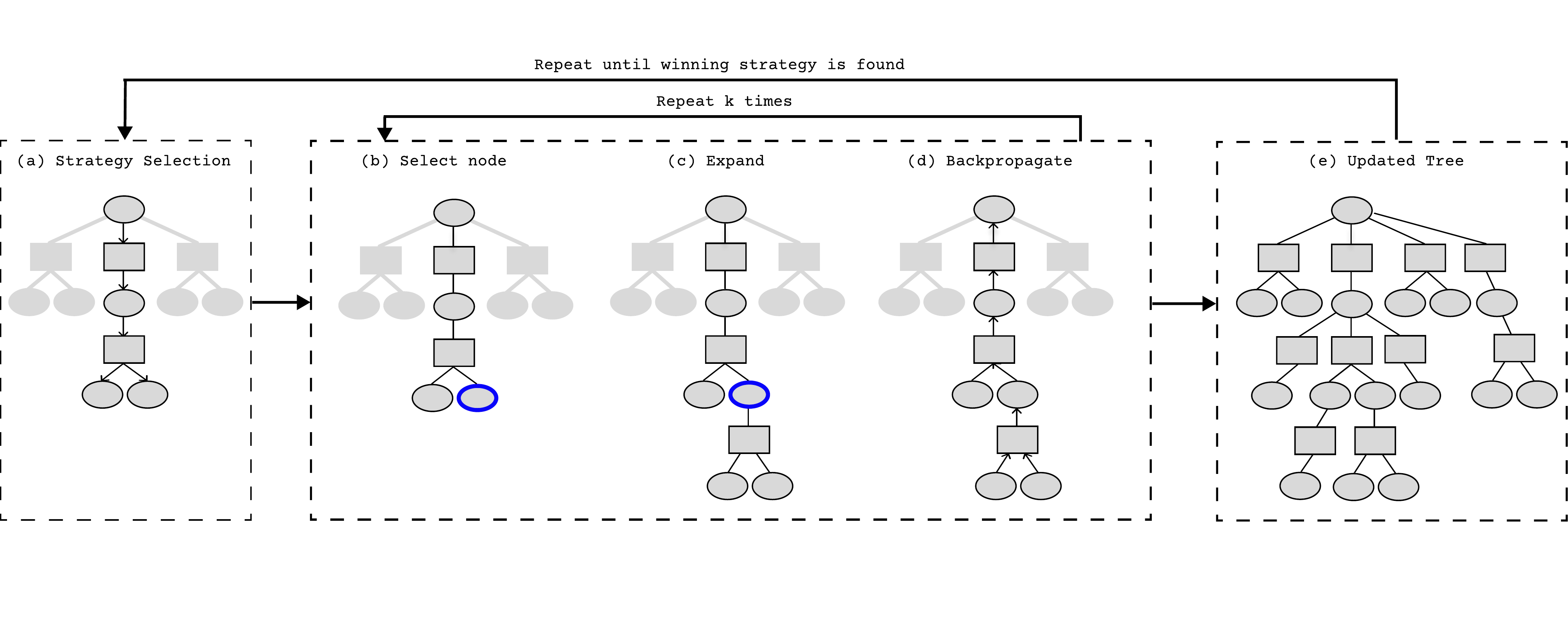}
    \caption{Illustration of SaBRS algorithm. Circles represent \texttt{OR} (\texttt{MIN} player) nodes, while squares represent \texttt{AND} (\texttt{MAX} player) nodes.}
    \label{fig:framework}
    \vspace{-3mm}
\end{figure*}
\section{Background}

\subsection{Game Trees, \texttt{AND/OR} trees, and Strategies}
To approach Problem~\ref{prob:winning strategy}, we use the concept of \textit{game trees}. A game tree is a tree $\tree$ whose nodes and edges represent game venue positions and game moves, respectively~\cite{AIbook}. At each node, a set of inputs are available. Each node-input pair results in a set of children in the tree. For a tree $\tree = (\N,\E)$ with a set of nodes $\N$ and edges $\E$, we denote by $\E(n)$ the set of child nodes of $n \in \N$.

An \texttt{AND/OR} tree $\andortree$ models a game tree as a two-player min-max game\cite{AIbook}. The players are \texttt{MIN} and \texttt{MAX}.
The position resulting from \texttt{MIN} and \texttt{MAX} moves are represented in the tree by \texttt{OR} and \texttt{AND} nodes, respectively. Moves of the game proceed in strict alternation between \texttt{MIN} and \texttt{MAX} until no further moves are allowed by the rules of the game. After the last move, \texttt{MIN} receives a cost which is a function of the final position. The objective of \texttt{MIN} is to minimize the cost, while \texttt{MAX}'s goal is to maximize the cost.

\begin{definition}[Subtree]
     $\subtree = (\N_s,\E_s)$ is a \emph{subtree} of an \texttt{AND/OR} tree $\andortree = (\N,\E)$ if the following conditions hold:
     \begin{itemize}
         \item $\N_s \subseteq \N$ and $\E_s \subseteq \E$,
         \item $n \in \N_s$ is the root of $\subtree$ if $n$ is the root of $\andortree$,
         \item $|\E_s(n)| = 1$ if $n$ is an \texttt{OR} node and $\E(n) \neq \emptyset$, i.e., only one move of the \texttt{MIN} player is available in $\subtree$,
         \item $\E_s(n) = \E(n)$ if $n$ is an \texttt{AND} node, i.e., all the moves of the \texttt{MAX} player are available in $\subtree$.
     \end{itemize}
\end{definition}

\begin{definition}[Strategy]
    A strategy over a game tree is a mapping from a node to an element of the input set available at the node. A strategy can be represented as a subtree of an \texttt{AND/OR} tree. We refer to this as a \emph{strategy subtree}.
\end{definition}

In a \emph{reachability} game, the objective is to reach a set of target positions $T \subseteq \N$. Then, after the last move, the \texttt{MIN} player (root node of the \texttt{AND/OR} tree) is penalized for having leaf nodes outside $T$ through the cost function. For a given strategy subtree, if all the leaf nodes are in $T$, the root gets a zero cost; otherwise, the cost is strictly positive.

\begin{definition}[Winning Strategy]
    A strategy over a game tree is called \emph{winning} if the root node of its \texttt{AND/OR} tree representation has a cost of zero.
\end{definition}

\subsection{Multi-Armed Bandits and Upper Confidence Bounds}

In the classic multi-armed bandit decision-making problem, an agent is presented with multiple arms (actions). When the agent chooses an arm (action), it receives a penalty according to a function that is hidden from the agent. The agent's goal is to minimize cost. Since the agent does not know how the cost is generated, it needs to trade-off between \textit{exploration} of actions and \textit{exploitation} of the action that seems the least costly.

The agent adopts a policy on how to choose actions. The regret of a policy  is the loss caused by not always choosing the best action. As the agent pulls more arms over time, it gains more information; hence, changing the regret of its policy over time.
A policy is said to \textit{resolve} the exploration-exploitation tradeoff if its regret growth rate is within a constant factor of a theoretical lower bound \cite{finiteregret}. 

A theoretically sound algorithm whose finite-time regret is well-studied is the \emph{Upper Confidence Bound} (UCB1) algorithm \cite{finiteregret}.  UCB1 resolves the exploration-explotation tradeoff by using a bias term.  Specifically, the UCB1 defines a cost for choosing action $a$ according to
\begin{align}
    \label{eq:ucb1}
    \UCB(a) = \hat{C}(a) - e \sqrt{2 \, \ln(N) / {N_a}},
\end{align}
where $\hat{C}(a)$ is the current estimated cost of taking action $a$, $N$ is the number of trials, and $N_a$ is the number of times action $a$ has been chosen. Exploration constant $e$ is a non-negative constant that determines the relative ratio of exploration to exploitation, with $e = 1$ in \cite{finiteregret}. Term $\sqrt{2\ln(N) / N_a}$ is known as a \emph{bias} sequence that controls the probability that $a$ is chosen as $N$ increases. For real valued $\hat{C}(\cdot)$, as $N$ increases, every action $a$ is eventually taken since the bias term becomes a very large value.
\section{Sampling-based Bandit-guided Reactive Synthesis (\SaBRS) Algorithm}
\label{sec:methodology}
In this section, we present our novel approach to Problem~\ref{prob:winning strategy}. Our method is based on modeling the NHS as a two-player game, where the nondeterminism is an adversarial player whose objective is to prevent the robot player from achieving the temporal and reachability goals. The game venue is the hybrid state space. We explore this game venue by combining sampling-based techniques with game-theoretic approaches. Specifically, we use a bandit-guided strategy selection method to decide on which parts of the game venue to explore, and use sampling for exploration.

Our algorithm iteratively grows a search tree in the hybrid state space of $H$. This is done  by growing a search tree $\tree = (\N,\E)$ rooted at initial state $s_0$ and extending the tree based on the semantics of $H$ (details in Sec. \ref{sec:tree extension}). Each node $n \in \N$  of the tree is a tuple $n = \langle s, N, \mathbf{u} \rangle$, where $s$ is a hybrid state, $N$ is the number of times that node $n$ has been visited, and $\mathbf{u}$ is the set of piecewise constant control inputs previously selected at $n$. 
We model this tree as an \texttt{AND/OR} tree, where the robot is the \texttt{MIN} player which chooses a control $(u,t) \in n.\mathbf{u}$ at node $n$, followed by the \texttt{MAX} player, which chooses a child of the node-control pair $(n, (u,t))$ with the highest cost. Intuitively, the \texttt{MAX} player models the nondeterminism in $H$. Then, a strategy subtree on the \texttt{AND/OR} tree is a (piecewise constant) control strategy in the hybrid space. The goal of \texttt{MIN} is to find a winning strategy subtree, which has all of its leaf nodes in $S_{goal}$. Computing a winning strategy on the game tree is therefore equivalent to synthesizing a (piecewise constant) winning control strategy, solving Problem~\ref{prob:winning strategy}. A major challenge is how to effectively extend the game tree in the game venue (exploration) such that it turns a ``promising'' strategy subtree to a winning strategy (exploitation), if one exists.

To construct this game tree efficiently, we propose a reactive synthesis algorithm that consists of two main components: strategy subtree selection and strategy expansion/improvement. 
The algorithm, called \textit{Sampling-based Bandit-guided Reactive Synthesis} (\SaBRS), is shown in Alg.~\ref{alg:framework} and depicted in Fig.~\ref{fig:framework}. 
\SaBRS uses a novel bandit-based action selection methodology to select a strategy in the \texttt{AND/OR} tree for expansion. Then, the strategy expansion component uses random sampling to grow this selected subtree. 
This alternation of selecting promising strategy subtrees and expanding on them combines the exploration-exploitation properties of bandit-based techniques at the strategy-level with the effectiveness of sampling-based motion planners, resulting in an efficient and probabilistically complete algorithm (see Sec.~\ref{sec:analysis}).
The algorithm terminates when a winning strategy is found.
 \setlength{\textfloatsep}{0.75pt}
 \setlength{\intextsep}{0.75pt}
\begin{algorithm}[ht]
    \caption{\SaBRS Algorithm}
    \label{alg:framework}
    \SetKwInOut{Input}{Input}\SetKwInOut{Output}{Output}
    \Input{NHS $H$, Initial state $s_0$, Expansion ratio $k$, Exploration constant $e$, Time limit $T_{max}$}
    \Output{Winning Reactive Strategy $\pi^*$ from $s_0$}
    $n_0 \gets \langle s_0, 0, \emptyset \rangle$\\
    $\tree = (\mathcal{N} \gets \{n_0 \}, \mathcal{E} \gets \emptyset)$\\
    \While{$\underset{(u,t) \in n.u}{\min}Q(n_0, (u,t)) > 0$ and time $< T_{max}$}
    {
        $\strategytree \gets \texttt{UCB-ST}(n_0, e)$\\
    \For{$j = 1 \rightarrow k$}
    {
        $\strategytree \gets \texttt{Explore}(\strategytree)$ 
    }
    $\tree \gets \tree \cup \strategytree$
    }
    $\strategytree^* \gets \texttt{UCB-ST}(n_0, 0)$\\
    \Return $\pi^*$
\end{algorithm}

\subsection{Strategy Tree Selection with Upper Confidence Bounds}

Since the search tree can become extremely large, we wish to bias our search towards the strategies that are promising, i.e., close to a winning strategy. Hence, in each iteration of planning, the algorithm evaluates and selects a strategy subtree $\strategytree$ of the search tree as depicted in the Frame (a) of Fig.~\ref{fig:framework}. 
The evaluation of a strategy subtree $\pi$ is done based on a cost function that assigns to node $n$ the cost
\begin{align}
    \label{eq:costfunction}
    C^\pi(n) = 1 - \frac{|\mathcal{C}_{goal}^\pi(n)|}{|\mathcal{C}_{all}^\pi(n)|} , 
\end{align}
where $\mathcal{C}_{all}^\pi(n)$ is the set of all leaf nodes of the subtree $\pi$ rooted at $n$, and 
$\mathcal{C}^\pi_{goal}(n) \subseteq \mathcal{C}_{all}^\pi(n)$ is the subset of nodes that are in $S_{goal}$. Intuitively, the cost of $n$ under strategy $\pi$ is the portion of leaf nodes that are not in $S_{goal}$. When $C^\pi(n) = 1$, no branches of subtree $\pi$ from $n$ end in $S_{goal}$, and when $C^\pi(n) = 0$, all the branches of the subtree end in $S_{goal}$ from $n$, i.e., $\pi$ is a winning strategy for $n$.

\begin{remark}
    It is important to note that various cost functions are possible in this framework, and the effectiveness of a cost function may be problem dependent. The only requirement of a cost function is that $C^\pi(n)=0$ if $\pi$ is a winning strategy from $n$, and $C^\pi(n)>0$ otherwise.
\end{remark}

From \eqref{eq:costfunction}, we define \emph{$\Q$-cost} to be the cost for choosing an input $(u,t)$ at node $n$ and then following strategy $\pi$ at subsequent nodes, i.e.,
\begin{align}
    \label{eq:qcost}
    \Q^\pi(n, (u,t)) = 1 - \frac{\sum_{n' \in children(n,(u,t))}|\mathcal{C}_{goal}^\pi(n')|}{\sum_{n' \in children(n,(u,t))}|\mathcal{C}_{all}^\pi(n')|}.
\end{align}
This $\Q$-cost allows us to evaluate the relative cost of each input $(u,t)$ at node $n$. Note that when $\pi(n) = (u,t)$, the $\Q$-cost is equivalent to $C^\pi(n)$, i.e., $\Q^\pi(n,\pi(n)) = C^\pi(n)$. The minimization of $\Q$-cost at each node thus provides us with strategies that are seemingly closer to a winning strategy subtree. However, this leads us to the classical \textit{exploitation-exploration dilemma}, since a strategy subtree with low cost may not always be the best strategy to choose, since a strategy may have a low cost but it may be difficult or impossible to be extended into a winning strategy subtree due to, e.g., one of its leaf nodes being stuck in a ``dead end". 

Therefore, we choose a strategy subtree by treating the control selection problem at each node as a separate multi-armed bandit problem to solve this exploration-exploitation tradeoff. To this end, we propose an adaptation of the UCB1 algorithm to be used in the context of strategy subtree selection. We call this new algorithm \emph{UCB for Strategy Tree selection} (\texttt{UCB-ST}). The algorithm is shown in Alg.~\ref{alg:UCB-ST}. It first initializes the chosen strategy tree $\strategytree$ with the root node $n_0$ (Line $4$ Alg.~\ref{alg:framework}). 
From $n_0$, it selects control inputs $(u,t)_{sel}$ in the \texttt{AND/OR} tree according to the UCB1 criterion (Line 3 in Alg.~\ref{alg:UCB-ST}) adapted from \eqref{eq:ucb1}:
\begin{align}
    \label{eq:UCB1}
    \UCB(n, (u,t)) = \Q^*(n, (u,t)) - e\sqrt{ 2 \ln (n.N) / n.N_{(u,t)} },
\end{align}
where $\Q^*(n,(u,t)) = \min_\pi \Q^\pi(n,(u,t))$ is the optimal $\Q$-cost, and $n.N_{(u,t)}$ is the number of times the control input $(u,t)$ has been selected at node $n$. All children nodes of $(n,(u,t)_{sel})$ are added to $\strategytree$, and control inputs for each children are again selected according to \eqref{eq:UCB1} (Lines 5-6). This process repeats until a strategy subtree $\pi$ of $\tree$ is obtained, which is when all the leaf nodes of a subtree are reached. \texttt{UCB-ST} allows the tree to grow in more promising parts of the tree, while still allowing for exploration of parts that seem less promising but might eventually lead to a winning strategy.

\begin{algorithm}[ht]
    \caption{\texttt{UCB-ST}($n, e$)}
    \label{alg:UCB-ST}
     $\pi \gets \{n\}$; $n.N = n.N + 1$\\
    \If{$n$ is not a leaf node}
    {
        $(u,t)_{sel} \gets \argmin_{(u, t) \in n.u} \UCB(n, (u,t))$ using \eqref{eq:UCB1} \\
        $\pi(n) = (u,t)_{sel}$\\
        \For{$n' \in children(n, (u,t)_{sel})$}
        {
            $\strategytree \gets \strategytree \cup \texttt{UCB-ST}(n', e)$
        }
    }
    \Return $\pi$
\end{algorithm}
\subsection{Strategy Improvement with Sampling-based Expansion}
\label{sec:tree extension}

A strategy subtree $\strategytree$ is extended in a sampling-based tree expansion manner, by growing the tree in the hybrid state space. This sampling-based expansion technique is inspired by motion planning algorithms. Pseudocode for our exploration algorithm is shown in Alg.~\ref{alg:motionplanner} and depicted in Frames (b)-(d) of Fig.~\ref{fig:framework}. In each iteration of exploration, a node $n$ in $\strategytree$ that has non-zero cost is first randomly sampled. Note that zero cost nodes already have a winning subtree, and do not need to be expanded further. Let $n.s = (q,x)$. Then, a control $u \in U_q$ and time duration $t$ are randomly sampled, and the node's continuous state $x$ is propagated by $F_{q}$. Any tree-based sampling-based planner that supports kinodynamic constraints (e.g., RRT or EST~\cite{Kingston2018samplingbased}) can be used in this step.

During propagation, the invariant $I_q$ checks the validity of the generated trajectory, and reachability indicator $R$ checks if the trajectory visits $S_{goal}$.  If a guard $G_{q \modesetprime}$ is enabled during propagation at continuous state $x'$, propagation is terminated and, for every $q' \in \modesetprime$,
node $n' = \langle (q',J_{qq'}(x')), 0, \emptyset \rangle$ is created.  If no guard is triggered and $I_q$ remains true for the entire duration $t$, only one new node is created. Then, the control $(u,t)$ is added to the set of sampled controls $n.\mathbf{u}$, and the new leaf nodes are added to the tree. Finally, the cost of nodes in the strategy is updated by backpropagation using \eqref{eq:costfunction}.

This expansion step is repeated $k$ times for each strategy subtree selection iteration to ensure that sufficient exploration is performed for a strategy subtree. Note that $k$ is an input to \SaBRS. The choice of $k$ can affect the efficiency of the algorithm. When $k$ is small, the algorithm quickly switches between promising strategy subtrees but may not expand a strategy subtree sufficiently. On the other hand, when $k$ is large, the algorithm expands a strategy subtree extensively, but switching between strategy subtrees becomes slower.

\begin{remark}
    \SaBRS also works in an anytime fashion, i.e., when given a time limit, \SaBRS returns a control strategy that minimizes root node cost. This also allows us to find partial solutions for problems, in which no winning strategy exists.
\end{remark}
\setlength{\textfloatsep}{2pt}
\begin{algorithm}[t]
    \SetKwInOut{Input}{Input}
    \SetKwInOut{Output}{Output}
    \caption{\texttt{Explore}($\strategytree$)}
    \label{alg:motionplanner}
    $n_{select} \leftarrow$ SampleAndSelect($\strategytree, s_{rand}$)\\
    $u_{rand}, t_{rand} \leftarrow$ SampleControl\&Dur($U_{n.s.q}, (0, T_{prop}]$)\\
    $\mathrm{N}_{new} \leftarrow$ Propagate($n_{select}, u_{rand}, t_{rand}$)\\
    \For{$n_{new} \in  \mathrm{N}_{new}$}
    {
        \uIf {IsValidTrajectory($n_{select}, n_{new}$)}{
            $n_{select}.\mathbf{u} \gets n_{select}.\mathbf{u} \cup \{(u_{rand}, t_{rand})\}$\\
            Add vertex and edge to $\strategytree$\\
            Update costs in $\pi$ by backpropagation
        }
    }
    \Return $\strategytree$
\end{algorithm}
\section{Analysis}
\label{sec:analysis}
In this section, we prove probabilistic completness of our algorithm. Specifically, we consider the case that kinodynamic RRT \cite{kinoRRT} is used as the strategy expansion technique. We begin by defining the notion of probabilistic completeness for algorithms that solve Problem~\ref{prob:winning strategy}.

\begin{definition}[Probabilistic Completeness]
    Given an NHS $H$ as in Def.~\ref{def: NHS}, 
    an algorithm is probabilistically complete if, as the number of iterations $K \rightarrow \infty$, the probability of finding a winning control strategy, if one exists, approaches $1$.
\end{definition}

Next, we prove that our strategy selection methodology repeatedly selects every strategy.
\begin{lemma}
    \label{lemma:UCB-ST}
    Given a search tree $\tree$ and exploration constant $e > 0$, \texttt{UCB-ST} in Alg.~\ref{alg:UCB-ST} selects every strategy subtree infinitely often, i.e., as number of iterations approaches infinity, the number of times every strategy subtree of $\tree$ is selected also approaches infinity.
\end{lemma}
\begin{proof}
    Consider a node $n \in \tree$. $\tree$ has a finite number of nodes for a finite number of SaBRS iterations. Input $(u,t) \in n.\mathbf{u}$ is selected according to UCB1 criterion \eqref{eq:UCB1}, which weighs the current $\Q$-cost with the exploration term $e\sqrt{2\ln (n.N)/ n.N_{(u,t)}}$. The exploration term increases if $(u,t)$ is not selected. As $N$ increases, another input $(u,t)' \neq (u,t)$ is always selected only if $\Q$-cost of $(u,t)'$ decreases at a rate faster than the increase in the exploration term of $(u,t)$. However, the cost function is defined such that the minimum Q-cost is $0$, and therefore, for $e > 0$, the exploration term for $(u,t)$ eventually dominates the $\Q$-cost terms. Since the search tree is finite, $|n.\mathbf{u}|$ is finite, and given sufficient iterations, each $(u,t) \in n.\mathbf{u}$ is eventually selected. By induction, every strategy subtree is selected infinitely often.
\end{proof}

Similar to continuous space RRT, we require that a solution with a non-zero radius of clearance exists, as defined below.

\begin{definition}[Clearance]
    For strategy $\pi$, let $Traj^\pi$ be the trajectories of $H$ induced under $\pi$ and $S^\pi = X^\pi \times Q^\pi \subset S$ be the set of hybrid states visited by $Traj^\pi$.  Further, denote by $\mathbb{B}_{\delta}(x)$ the ball centered at point $x$ with radius $\delta \geq 0$.  
    Clearance $\delta_{clear}^\pi$ of $\pi$ is the supremum of radius $\delta$
    such that for every $s = (x,q) \in S^\pi$ and $\forall x' \in \mathbb{B}_{\delta}(x) \subset \mathbb{R}^{n_q}$, it holds that $I_q(x') = I_{q}(x)$ and $G_{qQ'}(x') = G_{qQ'}(x)$.
\end{definition}
\noindent 
Intuitively, $\delta^\pi_{clear}$ defines a tube around each $Traj^\pi$ that contains all hybrid trajectories that follow the same sequence of modes and their continuous components remain $\delta^\pi_{clear}$ close.

We now formally state the main result of our analysis, which is that SaBRS (Alg.~\ref{alg:framework}) is probabilistically complete.

\begin{theorem}[Probabilistic Completeness]
    Alg.~\ref{alg:framework} is probabilistically complete if there exists a winning strategy $\strategytree^*$ with clearance $\delta^{\pi^*}_{clear} > 0$.
\end{theorem}

\begin{proof}
    Consider a winning strategy subtree $\strategytree_{win}$ from $s_0$ to $S_{goal}$ with clearance $\delta_{clear} > 0$. Let $P$ be the set of all paths from $s_0$ to a leaf in $S_{goal}$. Now, a path $p_i \in P$ can be described by the sequence of nodes with states $p_i = s_0^{g_0} \ldots s^{g_1}_{k} \ldots s^{g_0}_{l} \ldots s^{g_1}_{m} \ldots s_{goal, p_i}$ ending in a node $s_{goal} \in S_{goal}$, where the superscript $g_1$ denotes that a guard is triggered, and $g_0$ denotes the guard is not triggered. Cover $p_i$ with a set of balls of radius $\delta$ centered at $s_0^{g_0}, s_1, \cdots, s_{goal, p_i}$. We say that a path $p_j$ \emph{follows} another path $p_i$ if each vertex of $p_j$ is within the $\delta$ radius ball of $p_i$.

    Since the flow functions $F$ and jump functions $J$ are Lipschitz continuous (Assumption~\ref{assumption: Lipschitz}), from \cite[Theorem 2]{rrtpc}, we are guaranteed that RRT asymptotically almost surely finds a control trajectory from $s_0$ to $S_{goal}$ that follows $p_i$ when starting from a tree which contains $s_0$. From Lemma~\ref{lemma:UCB-ST}, \texttt{UCB-ST} will always eventually select any strategy subtree $\strategytree_i$ of our search tree $\tree$. Let $t$ be the number of paths in $\strategytree$ that uniquely follows a path $p_i \in P$. Assume that at step $j$, the selected subtree $\strategytree$ contains $0 < t <|P|$ paths. As expansion iterations increase, $\strategytree$ asymptotically almost surely finds a path from $s_0$ to $S_{goal}$ that follows a new path $p_l \in P$ which it did not uniquely follow before. Hence, the new $\strategytree^+_{i}$ expanded from $\strategytree_{i}$ now contains $t + 1$ paths that unique follows paths in $P$. From Lemma~\ref{lemma:UCB-ST}, $\strategytree^+_i$ will eventually be selected again. By induction, $t \rightarrow |P|$ and a winning strategy is found.
\end{proof}
\section{Extensions to Improve Base Algorithm}
\label{sec:extensions}
We present three extensions that can improve efficiency of SaBRS without affecting its probabilistic completeness.

 \setlength{\textfloatsep}{12pt}
 \setlength{\intextsep}{12pt}
\begin{table*}[!ht]
    \centering
    \caption{ Benchmark planner performance results. We report the mean time with standard error, and success rate over 100 simulation trials, with best scores in bold. RRT in case study 2 and MCTS for both cases are excluded from the table since they had $0$ success rate.}
    \label{tab:benchmarks}
    \scalebox{0.94}{
    \begin{tabular}{c c c c c c c c c c}
    \toprule
    & \multirow{2}{*}{Algorithm}&  \multicolumn{2}{c}{\underline{\hspace{7mm} Environment 1 \hspace{7mm}}} & \multicolumn{2}{c}{\underline{\hspace{7mm} Environment 2 \hspace{7mm}}} & \multicolumn{2}{c}{\underline{\hspace{7mm} Environment 3 \hspace{7mm}}} & \multicolumn{2}{c}{\underline{\hspace{7mm} Environment 4 \hspace{7mm}}}\\
    &  & Time (s)& Success  (\%) & 
    Time (s)& Success (\%) & Time (s)& Success  (\%) & Time  & Success  (\%)\\
    \midrule
     \multirow{3}{*}{Case 1} & RRT   & $299.0 \pm 0.0$ & $3$  & $299.5 \pm 0.0$ &  $2$  & $300 \pm 0.0$  &  $1$ & --  & $0$ \\
     & Planner in \cite{Lahijanian2014} & $78.1 \pm 4.6$ & $91$ & $113.1 \pm 9.1$ & 82 & $153.68 \pm 15.5$ & $44$ & $222.5 \pm 26.3$ & $10$ \\
     & SaBRS  (Ours) & $\mathbf{4.2 \pm 0.5}$ & $\mathbf{100}$ & $\mathbf{8.3 \pm 1.1}$ & $\mathbf{100}$ & $\mathbf{23.04 \pm 3.6}$ & $\mathbf{99}$ & $\mathbf{57.9 \pm 6.5}$ & $\mathbf{93}$ \\
     \midrule
     \multirow{2}{*}{Case 2}  & Planner in \cite{Lahijanian2014} & $98.1 \pm 6.5$ & 96 & $157.1 \pm 12.0$ & $68$ & $208.9 \pm 18.3$ & $32$ & $251.1 \pm 49.0$ & $4$ \\
    & SaBRS (Ours) & $\mathbf{4.6 \pm 0.5}$ & $\mathbf{100}$ & $\mathbf{11.4 \pm 2.8}$ & $\mathbf{99}$ & $\mathbf{24.4 \pm 4.0}$ & $\mathbf{97}$ & $\mathbf{48.99 \pm 6.8}$ & $\mathbf{88}$\\
    \bottomrule
    \end{tabular}
    }
    \vspace{-4.2mm}
\end{table*}

\paragraph*{Warm Starting}
The effectiveness of SaBRS relies on selecting promising strategies based on \eqref{eq:costfunction}. However, when no branch of the search tree is in $S_{goal}$ yet, the costs of strategies remain the same, namely $1$. 
To improve the effectiveness of strategy selection, we first perform a \textit{warm start} of SaBRS by performing sampling-based exploration on the full \texttt{AND/OR} tree for some fixed time, or until at least one leaf node is in $S_{goal}$. In our experiments, we observed that warm starting is especially useful in problems with longer horizons, where reaching a goal by a leaf node is difficult.

\paragraph*{Strategy Tree Expansion Heuristics}
During strategy expansion, SaBRS uses sampling-based exploration. This exploration has been shown to be greatly improved by heuristic guidance mechanisms, such as goal bias and trajectory bias \cite{heuristicbias}. The exploration step of our algorithm is general and amenable to such heuristics. An example of such a heuristic
is \emph{Guided Path-Generation}, introduced in \cite{Lahijanian2014} for Problem~\ref{prob:winning strategy}. It uses the search tree branches that end in a goal state to guide the expansion of nodes (see \cite{Lahijanian2014} for details). 
To maintain probabilistic completeness of our algorithm, we use Guided Path-Generation with low probability $p$ and the random exploration of Alg.~\ref{alg:motionplanner} with probability $1-p$. In our evaluations, we find that guided path generation greatly improves efficiency if the hybrid goal set has the same continuous component in the goal modes, i.e., $S_{goal} = Q_{goal} \times X_{goal}$, where $Q_{goal} \subseteq Q$ and $X_{goal} \subseteq \bigcap_{q \in Q_{goal}} X_q$.

\paragraph*{Sub-Strategy Tree Selection}
At each strategy selection iteration, Alg.~\ref{alg:UCB-ST} chooses a subtree $\strategytree$ of the \texttt{AND/OR} tree $\tree$. This allows us to expand promising strategies currently available in $\tree$. $\tree$ has a discrete set of available inputs, a subset of the continuous input space. During exploration, $\strategytree$ is extended, increasing the discrete set of available inputs. When the search tree is very large and deep, it may become difficult to improve \textit{smaller} strategies within a subtree. To ameliorate this issue, we probabilistically prune parts of a subtree to obtain a smaller strategy, with probability $\rho$. This effectively chooses `no action' at an \textit{OR} node, leading to a smaller strategy that may be improved to become a winning strategy more easily. In our experiments, we observe that this improves search for longer horizon problems where there are many leaf nodes that lead to ``dead ends".
\section{Experiments}
We evaluate the performance of \SaBRS against $3$ state-of-the-art algorithms, namely kinodynamic RRT, continuous space Monte Carlo Tree Search (MCTS-PW), and the motion planner in \cite{Lahijanian2014}, in a series of benchmarks. We also provide several illustrative examples to show the generality of \SaBRS. We implemented all algorithms in \texttt{C++} using OMPL \cite{sucan2012the-open-motion-planning-library}. All computations were performed single-threaded on a nominally 3.60 GHz CPU with 32 GB RAM. For \SaBRS, we used $k = 5000$ and $e = 0.0005$ for all planning instances.

\subsection{Benchmarking Results}
We first evaluate the algorithms on the benchmark problems proposed in \cite{Lahijanian2014}.  
The problems consist of two NHS models of a three-gear second-order car system that is subject to nondeterminism when shifting gears in four environments \cite{Lahijanian2014}. In Case Study $1$, the system, when shifting from gear two to three, may mistakenly shift to gear one. In Case Study $2$, there is an additional nondeterminism when the car shifts from gear three to one. In this case, the system may mistakenly change to gear two instead of one. A solution is a winning strategy that reaches the goal state over all possible gear transitions from the initial state. We refer the reader to \cite{Lahijanian2014} for details on the dynamics of each gear. 

We conducted $100$ trials for each of the four environments and two cases, and set a time limit of $300$ seconds to find a solution per trial. The results are summarized in Table~\ref{tab:benchmarks}. In Case Study $1$, MCTS did not find any solutions. In Case Study $2$, both RRT and MCTS did not find any solutions. It is evident from the poor performance of both RRT and MCTS that neither a pure sampling-based exploration method nor a pure heuristic search method is effective for finding winning strategies for NHS planning problems. Further, we see that \SaBRS significantly outperforms the compared methods both in computation time (up to an order of magnitude) and success rate (up to 3$\times$) for finding winning strategies. This suggests that the combination of the bandit-based game theoretic approach for strategy selection and sampling-based exploration is important for reactive synthesis under nondeterminism.

\subsection{Robotic Charging System}

Next, we showcase \SaBRS's versatility to handle NHS with time constraints. Consider a robot with second-order car dynamics with a bounded motion disturbance at every time step, which is equipped with a closed-loop controller that is able to maintain the robot in a ball of radius $r$ around its nominal plan. 
The robot is tasked with navigating to the charger within $2$ minutes. If it goes over a water puddle, the robot needs to dry off on the carpet for $10$ seconds before going to the charger. During online execution, the robot has perfect observation of its state. However, during offline planning, due to the motion disturbance, if the radius $r$ ball around a nominal plan intersects with a water puddle, the robot may traverse the puddle during execution. Hence, the robot motion can be modelled as a NHS.

Figs.~\ref{fig:waterpuddle1} and \ref{fig:waterpuddle2} show two examples with different obstacle configurations. In both cases, \SaBRS finds a solution within $30$ seconds. In Fig.~\ref{fig:waterpuddle1}, when there is more space for the robot to traverse in the dry regions (in white), \SaBRS synthesizes a strategy such that the entire ball stays within the white region without touching the water (in blue) and navigates to the charger (green). When the environment is more cluttered (Fig.~\ref{fig:waterpuddle2}), such that the robot cannot reach the charger without guaranteeing that it does not go over the water puddle due to the uncertain ball, \SaBRS plans a reactive strategy with two possible trajectories. If it does not get wet, it navigates directly to the charger. If it gets wet, it first goes to and stays on the carpet for $10$ seconds before navigating to the charger.
\begin{figure}[t]
    \centering
    \begin{subfigure}{0.22\textwidth}
    \includegraphics[width=\textwidth, trim={1.75cm 2.7cm 1.4cm 2.5cm},clip]{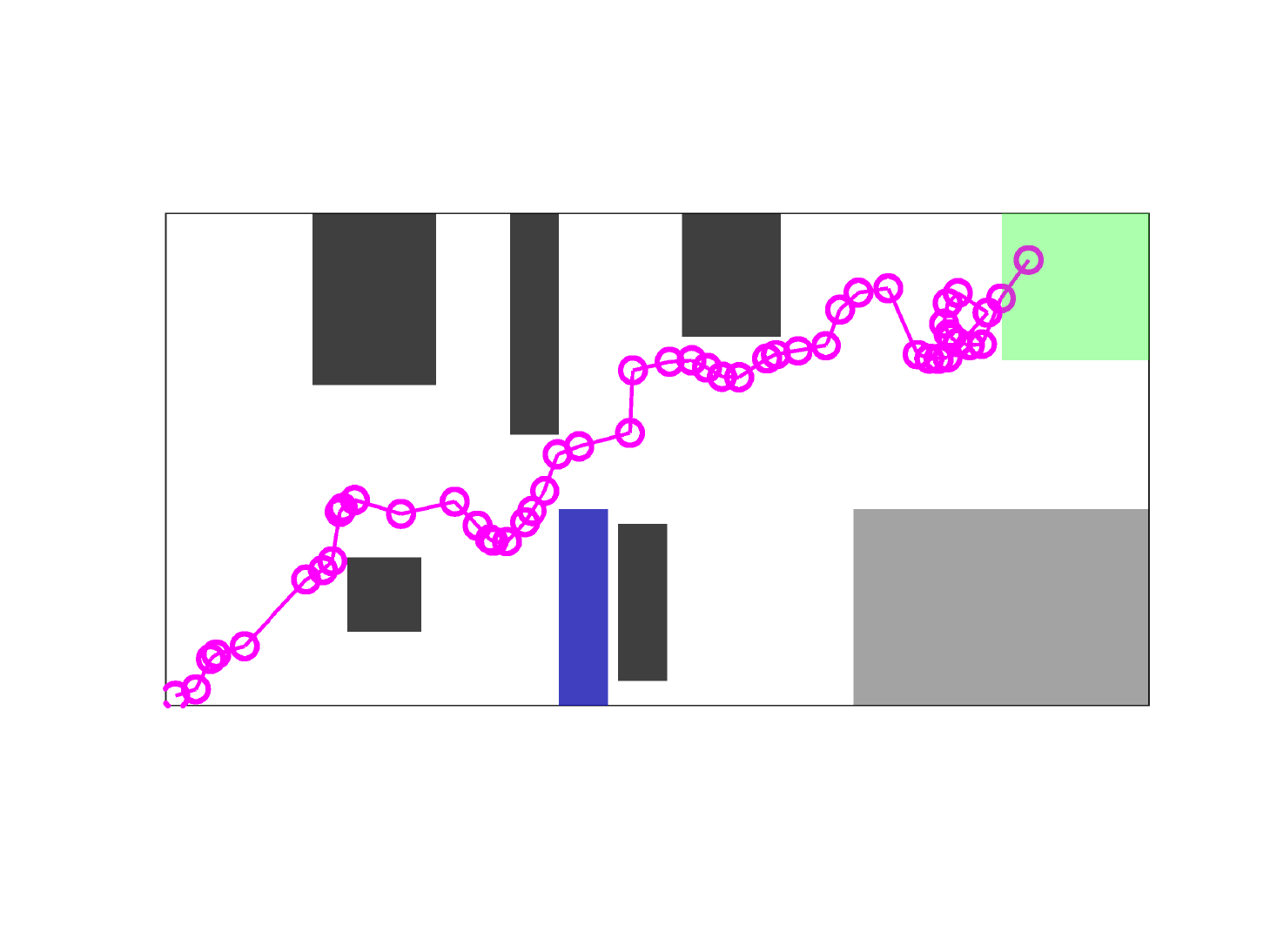}
        \vspace{-6mm}
    \caption{\footnotesize Spaced Obstacles}
    \label{fig:waterpuddle1}
    \end{subfigure}
    \begin{subfigure}{0.22\textwidth}
    \includegraphics[width=\textwidth, trim={1.75cm 2.7cm 1.4cm 2.5cm},clip]{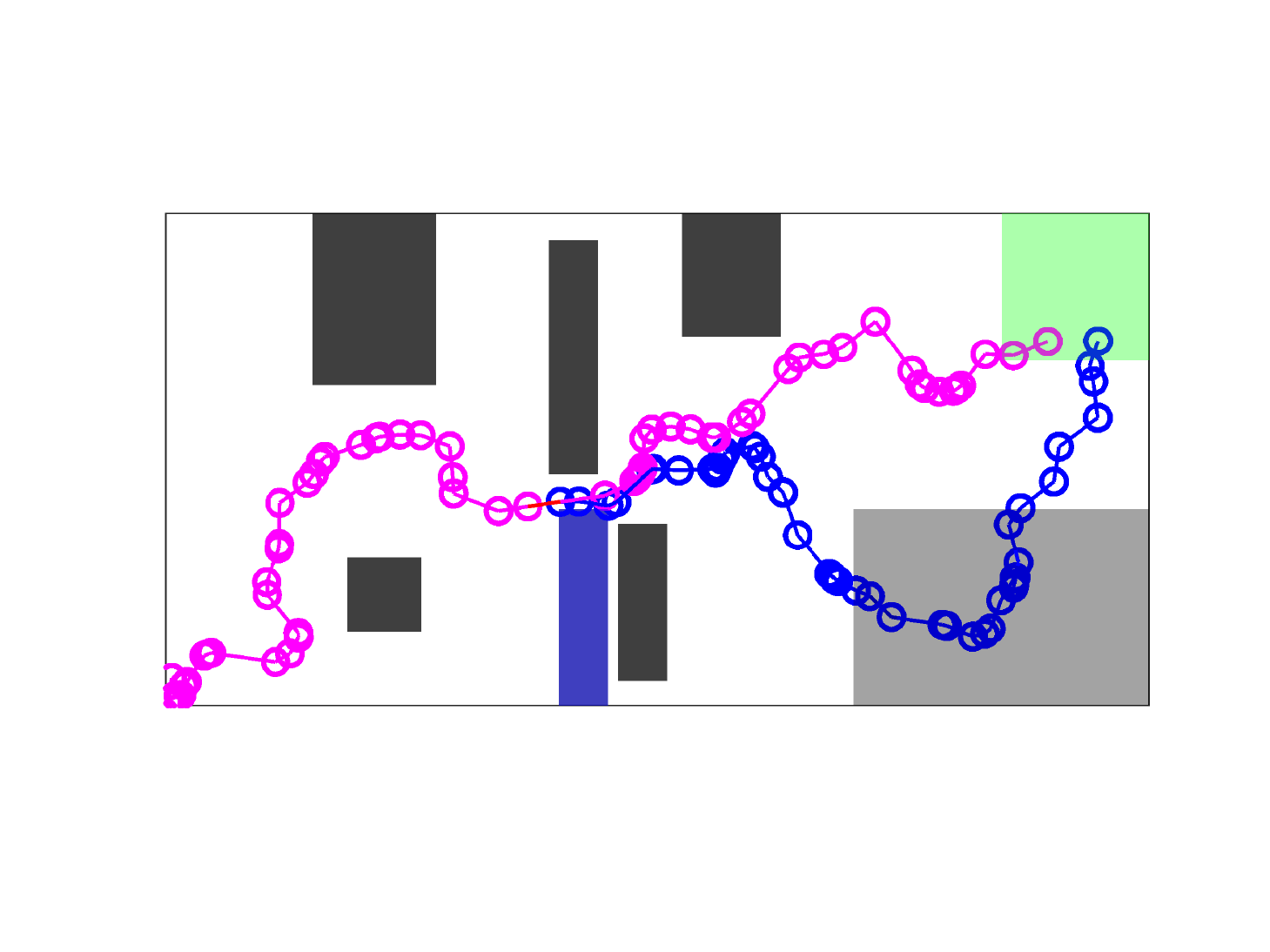}
    \vspace{-6mm}
    \caption{\footnotesize Cluttered Obstacles}
    \label{fig:waterpuddle2}
    \end{subfigure}
\caption{\small \SaBRS synthesizes strategies that account for nondeterminism in traversing the puddle due to motion uncertainty in the robotic charging example. The black, blue, grey and green regions represent obstacles, puddle, carpet, and charger, respectively.}
\label{fig:puddle}
\vspace{-3mm}
\end{figure}

\subsection{Search-and-rescue Scenario}

Next, we show that \SaBRS is effective for problems with complex temporal specifications. We consider variants of the search-and-rescue scenario of Example~\ref{ex:searchandrescue}, in which each room's door state is fixed but unknown at planning time. Fig.~\ref{fig:searchandrescue_trajectories} show example synthesized strategies (computed within $20s$) for a 3-room variant. Each possible trajectory in the reactive strategies is plotted with a different color, and a red dotted path segment indicates a nondeterministic mode transition. \SaBRS successfully finds winning strategies that reacts based on if a door is blocked, and if the human is in each room.

Finally, we compute and demonstrate strategies in real world experiments for a 2-room variant of the search-and-rescue scenario on a ClearPath Jackal robot. A demonstration video can be found at {\small\url{https://youtu.be/Hy1rNGfX6Zg}}.

\begin{figure}[t]
    \centering
    \begin{subfigure}{0.22\textwidth}
    \includegraphics[width=\linewidth, trim={14.75cm 2.75cm 12.5cm 2.25cm},clip]{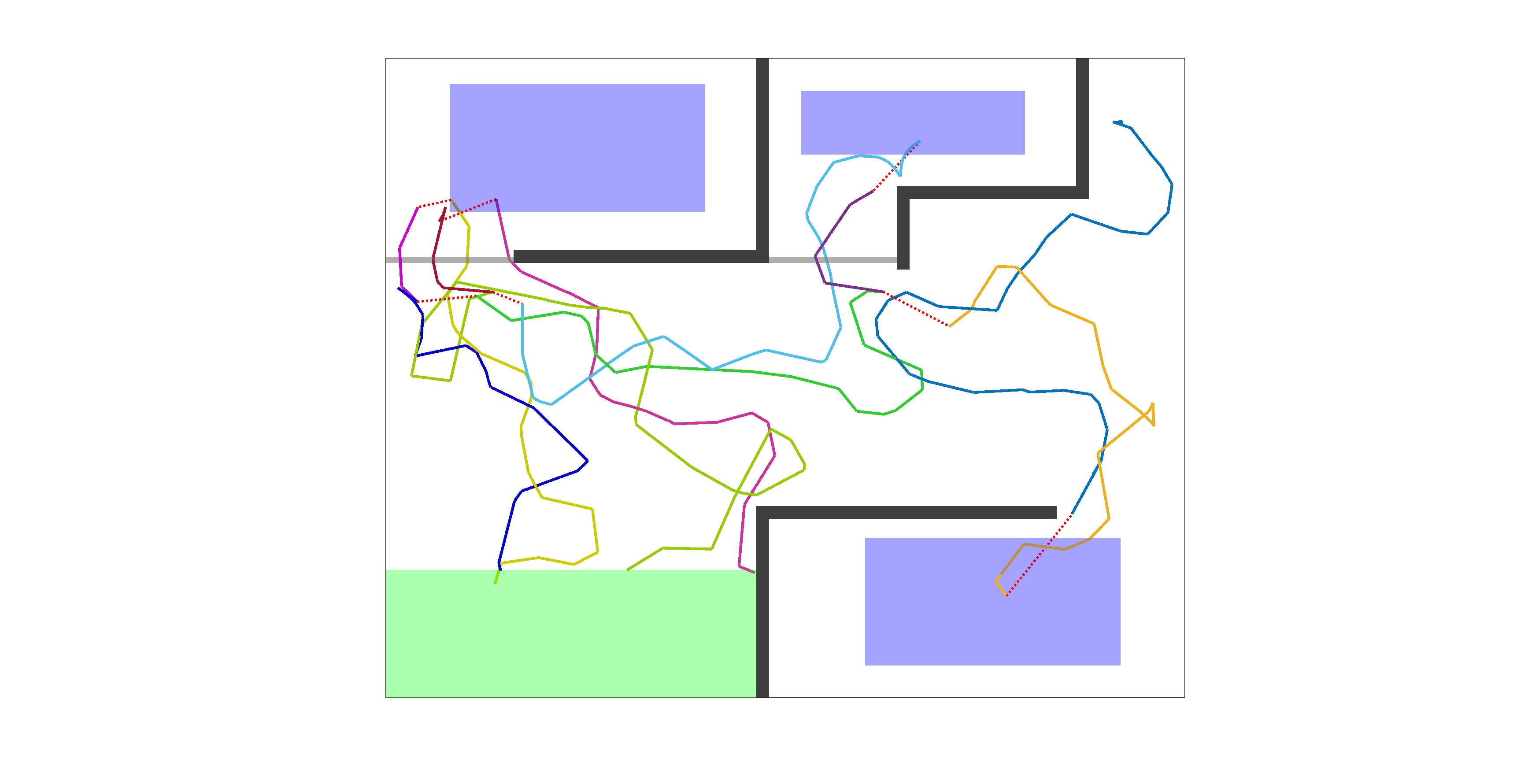}
        \vspace{-6mm}
    \caption{\footnotesize Human may be in all rooms.}
    \label{fig:rescuetrajectory1}
    \end{subfigure}
    ~~~
    \begin{subfigure}{0.22\textwidth}
    \includegraphics[width=\linewidth, trim={14.75cm 2.75cm 12.5cm 2.25cm},clip]{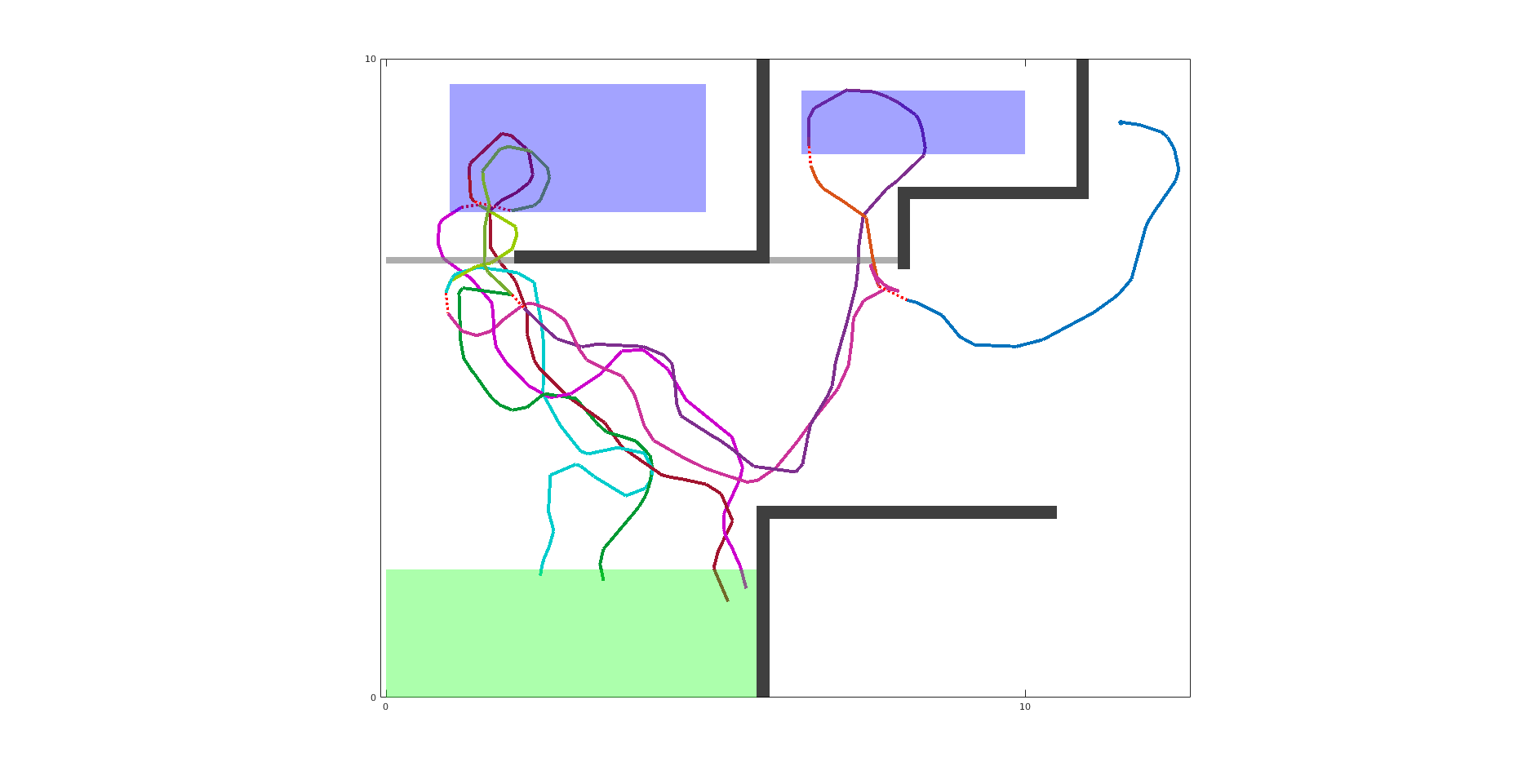}
    \vspace{-6mm}
    \caption{\footnotesize Human in either of 2 rooms.}
    \label{fig:rescuetrajectory2}
    \end{subfigure}
\caption{ \SaBRS synthesizes strategies in the search-and-rescue example. Purple regions represent possible human locations.}
\label{fig:searchandrescue_trajectories}
\vspace{-3mm}
\end{figure}
\section{Conclusion and Future Work}

This paper considers the problem of computing reactive strategies for NHS under temporal and reachability constraints. We propose \SaBRS, an algorithm that guarantees reachability under all possible NHS evolutions. \SaBRS combines the effectiveness of sampling-based motion planning with bandit-based techniques. We show that \SaBRS is probabilistically complete, and benchmarks and case studies demonstrate its effectiveness. For future work, we plan to extend \SaBRS to explicitly include uncertainty in continuous dynamic, and extend \SaBRS to also optimize an objective function.
\bibliographystyle{IEEEtran}
\bibliography{bib}

\end{document}